\documentclass[11pts]{article}
\usepackage{graphicx,amsmath,amssymb}
 \usepackage{hyperref}

\newtheorem{theorem}{Theorem}[section]
\newtheorem{lemma}[theorem]{Lemma}
\newtheorem{corollary}[theorem]{Corollary}

\newtheorem{kernelrule}{Step}

\newcommand{\blackslug}{\penalty 1000\hbox{
    \vrule height 8pt width .4pt\hskip -.4pt
    \vbox{\hrule width 8pt height .4pt\vskip -.4pt
          \vskip 8pt
      \vskip -.4pt\hrule width 8pt height .4pt}
    \hskip -3.9pt
    \vrule height 8pt width .4pt}}

\newenvironment{proof}{$\;$\newline \noindent {\sc Proof.}$\;\;\;$\rm}{}
\newcommand{\qed}{\hspace*{\fill}\blackslug}

\def\boxit#1{\vbox{\hrule\hbox{\vrule\kern4pt
  \vbox{\kern1pt#1\kern1pt}
\kern2pt\vrule}\hrule}}

\begin{document}

\def\boxit#1{\vbox{\hrule\hbox{\vrule\kern4pt
  \vbox{\kern1pt#1\kern1pt}
\kern2pt\vrule}\hrule}}

\title{ \bf Cluster Editing: \\Kernelization based on Edge Cuts
  \thanks{Supported in part by the US National Science Foundation
    under the Grants CCF-0830455 and CCF-0917288.  Work partially done
    when both authors were visiting Central South University,
    Changsha, China.}  }

\author{
  {\sc Yixin Cao} \hspace*{22mm} {\sc Jianer Chen}\\
  { Department of Computer Science and Engineering}\\
  Texas A\&M University\\
  {\tt \{yixin, chen\}@cse.tamu.edu}}
\maketitle

\begin{abstract}
  Kernelization algorithms for the {\sc cluster editing} problem have
  been a popular topic in the recent research in parameterized
  computation.  Thus far most kernelization algorithms for this
  problem are based on the concept of {\it critical cliques}. In this
  paper, we present new observations and new techniques for the study
  of kernelization algorithms for the {\sc cluster editing}
  problem. Our techniques are based on the study of the relationship
  between {\sc cluster editing} and graph edge-cuts. As an application,
  we present an ${\cal O}(n^2)$-time algorithm that constructs a $2k$
  kernel for the {\it weighted} version of the {\sc cluster editing}
  problem. Our result meets the best kernel size for the
  unweighted version for the {\sc cluster editing} problem, and
  significantly improves the previous best kernel of quadratic size
  for the weighted version of the problem.
\end{abstract}


\section{Introduction}

Errors are ubiquitous in most experiments, and we have to find out
the true information buried behind them, that is, to remove the
inconsistences in data of experiment results.  In most cases, we
want to make the data consistent with the least amount of
modifications, i.e., we assume the errors are not too much.  This is
an \emph{everyday} problem in real life.  Indeed, the problem has
been studied by researchers in different areas
\cite{bansal-04-correlation-clustering,shamir-04-graph-modification-problem}.
A graph theoretical formulation of the problem is called the {\sc
cluster editing} problem that seeks a collection of edge
insertion/deletion operations of minimum cost that transforms a
given graph into a union of disjoint cliques. The {\sc cluster
editing} problem has applications in many areas, including machine
learning \cite{bansal-04-correlation-clustering}, world wide web
\cite{dean-99-www}, data-minning \cite{berkhin-06-survey},
information retrieval \cite{hearst-96-ir}, and computational biology
\cite{chen-03-phylogenetic}.  The problem is also closely related to
another interesting and important problem in algorithmic research,
{\sc clustering aggregation} \cite{ailon-08-ranking-and-clustering},
which, given a set of clusterings on the same set of vertices, asks
for a single clustering that agrees as much as possible with the
input clusterings.

Let $G = (V, E)$ be an undirected graph, and let $V^2$ be the set of
all unordered pairs of vertices in $G$ (thus, for two vertices $v$ and
$w$, $\{v, w\}$ and $\{w, v\}$ will be regarded as the same pair).
Let $\pi: V^2 \mapsto \mathbb{N} \cup \{+\infty\}$ be a {\it weight
function}, where $\mathbb{N}$ is the set of positive integers.  The
{\it weight} of an edge $[v, w]$ in $G$ is defined to be $\pi(v,
w)$. If vertices $v$ and $w$ are not adjacent, and we add an edge
between $v$ and $w$, then we say that we {\it
  insert an edge $[v, w]$ of weight $\pi(v, w)$}.

The weighted {\sc cluster editing} problem is formally defined as
follows:
\begin{quote}
  (Weighted) {\sc cluster editing}: Given $(G, \pi, k)$, where $G =
  (V, E)$ is an undirected graph, $\pi: V^2 \mapsto \mathbb{N} \cup
  \{+\infty\}$ is a weight function, and $k$ is an integer, is it
  possible to transform $G$ into a union of disjoint cliques by edge
  deletions and/or edge insertions such that the weight sum of the
  inserted edges and deleted edges is bounded by $k$?
\end{quote}

The problem is NP-complete even in its unweighted version
\cite{shamir-04-graph-modification-problem}. Polynomial-time
approximation algorithms for the problem have been studied. The best
result is a randomized approximation algorithm of expected
approximation ratio $3$ by Ailon, Charikar, and Newman
\cite{ailon-08-ranking-and-clustering}, which was later derandomized
by van Zuylen and Williamson \cite{zuylen-09}.  The problem has also
been shown to be APX-hard by Charikar, Guruswami, and Wirth
\cite{charikar-04-approximate-clustering}.

Recently, some researchers have turned their attention to exact
solutions, and to the study of parameterized algorithms for the
problem. A closely related problem is to study {\it kernelization
  algorithms} for the problem, which, on an instance $(G, \pi, k)$ of
{\sc cluster editing}, produces an ``equivalent'' instance $(G', \pi,
k')$ such that $k' \leq f(k)$\footnote{$f(\cdot)$ is a computable
  function.} and that the {\it kernel size} (i.e., the number of
vertices in the graph $G'$) is small. For the unweighted version of
the problem (i.e., assuming that for each pair $v$ and $w$ of
vertices, $\pi(v, w) = 1$), Gramm et
al. \cite{gramm-kernel-cluster-editing} presented the first
parameterized algorithm running in time ${\cal O}(2.27^k + n^3)$ and a
kernelization algorithm that produces a kernel of ${\cal O}(k^2)$
vertices. This result was immediately improved by a successive
sequence of studies on kernelization algorithms that produce kernels
of size $24k$ \cite{fellows-kernel-cluster-editing}, of size $4k$
\cite{guo-kernel-cluster-editing} and of size $2k$
\cite{chen-kernel-cluster-editing}.  The $24k$ kernel was obtained via
{\it crown reduction}, while the later two results were both based on
the concept of simple series module ({\it critical clique}), which is
a restricted version of modular decomposition
\cite{CE-80-decomposition}.  Basically, these algorithms iteratively
construct the modular decomposition, find reducible simple series
modules and apply reduction rules on them, until there are no any
reducible modules found.

For the weighted version, to our best knowledge, the only non-trivial
result on kernelization is the quadratic kernel developed by B\"ocker
et al.~\cite{bocker-09-weighted-clustering}.

The main result of this paper is the following theorem:
\begin{theorem}\label{thm:2k}
  There is an ${\cal O}(n^2)$-time kernelization algorithm for the
  weighted {\sc cluster editing} problem that produces a kernel which
  contains at most $2k$ vertices.
\end{theorem}

Compared to all previous results, Theorem~\ref{thm:2k} is better not
only in kernel size and running time, but also more importantly in
conceptual simplicity.


A more general version of weighted {\sc cluster editing} problem is
defined with real weights, that is, the weight function $\pi$ is
replaced by $\pi': V^2 \mapsto \mathbb{R}_{\geq 1} \cup \{+\infty\}$
where $\mathbb{R}_{\geq 1}$ is the set of all real numbers larger
than or equal to $1$, and correspondingly $k$ becomes a positive
real number.  Our result also works for this version, in the same
running time, and with only a small relaxation in the consant of
kernel size.

\paragraph{Our contribution.} We report the first linear vertex kernel
with very small constant, for the weighted version of the {\sc
cluster editing} problem. Our contribution to this research
includes:
\begin{enumerate}
\vspace{-2mm}
\item the cutting lemmas (some of them are not used for our
  kernelization algorithm) are of potential use for future work
  on kernelizations and algorithms;
\vspace{-2mm}
\item both the idea and the process are very simple with efficient
  implementations that run in time ${\cal O}(n^2)$.  Indeed, we
  use only a single reduction rule, which works
  for both weighted and unweighted versions;
\vspace{-2mm}
\item the reduction processes to obtain the above results are
  independent of $k$, and therefore are more general and applicable.
\end{enumerate}

\section{Cutting Lemmas}\label{sec:pre}

In this paper, graphs are always undirected and simple. A graph is a
{\it complete graph} if each pair of vertices are connected by an
edge. A {\it clique} in a graph $G$ is a subgraph $G'$ of $G$ such
that $G'$ is a complete graph. By definition, a clique of $h$
vertices contains ${h \choose 2} = h(h-1)/2$ edges. If two vertices
$v$ and $w$ are not adjacent, then we say that the edge $[v, w]$ is
{\it missing}, and call the pair $\{v, w\}$ an {\it anti-edge}.  The
total number of anti-edges in a graph of $n$ vertices is $n(n-1)/2 -
|E(G)|$. The subgraph of the graph $G$ induced by a vertex subset
$X$ is denoted by $G[X]$.

Let $G = (V, E)$ be a graph, and let $S \subseteq V^2$. Denote by $G
\triangle S$ the graph obtained from $G$ as follows: for each pair
$\{v, w\}$ in $S$, if $[v, w]$ is an edge in $G$, then remove the
edge $[v, w]$ in the graph, while if $\{v, w\}$ is an anti-edge,
then insert the edge $[v, w]$ into the graph. A set $S \subseteq
V^2$ is a {\it solution} to a graph $G = (V, E)$ if the graph $G
\triangle S$ is a union of disjoint cliques.

For an instance $(G, \pi, k)$ of {\sc cluster editing}, where $G =
(V, E)$, the {\it weight} of a set $S \subseteq V^2$ is defined as
$\pi(S) = \sum_{\{v, w\} \in S} \pi(v, w)$. Similarly, for a set
$E'$ of edges in $G$, the {\it weight} of $E'$ is $\pi(E') =
\sum_{[v, w] \in E'} \pi(v, w)$. Therefore, the instance $(G, \pi,
k)$ asks if there is a solution to $G$ whose weight is bounded by
$k$.

For a vertex $v$, denote by $N(v)$ the set of neighbors of $v$, and
let $N[v] = N(v) \cup \{v\}$.  For a vertex set $X$, $N[X] =
\bigcup_{v \in X} N[v]$, and $N(X) = N[X] \backslash X$. For the
vertex set $X$, define $\overline{X} = V \backslash X$. For two
vertex subsets $X$ and $Y$, denote by $E(X, Y)$ the set of edges
that has one end in $X$ and the other end in $Y$. For a vertex
subset $X$, the edge set $E(X, \overline{X})$ is called the {\it cut
of $X$}. The total cost of the cut of $X$ is denoted by $\gamma(X) =
\pi( E(X, \overline{X}) )$. Obviously, $\gamma(X) =
\gamma(\overline{X})$. For an instance $(G, \pi, k)$ of the {\sc
cluster editing} problem, denote by $\omega(G)$ the weight of an
optimal (i.e., minimum weighted) solution to the graph $G$.

Behind all of the following lemmas is a very simple observation: in
the objective graph $G \triangle S$ for any solution $S$ to the
graph $G$, each induced subgraph is also a union of disjoint
cliques. Therefore, a solution $S$ to the graph $G$ restricted to an
induced subgraph $G'$ of $G$ (i.e., the pairs of $S$ in which both
vertices are in $G'$) is also a solution to the subgraph $G'$. This
observation leads to the following {\it Cutting Lemma}.

\begin{lemma}\label{lem:cutting}
  Let ${\cal P} = \{V_1, V_2, \dots, V_p \}$ be a vertex partition of
  a graph $G$, and let $E_{\cal P}$ be the set of edges whose two
  ends belong to two different parts in $\cal P$. Then $\sum_{i=1}^p
  \omega(G[V_i]) \leq \omega(G) \leq \pi(E_{\cal P}) + \sum_{i=1}^p
  \omega(G[V_i])$.
\end{lemma}

\begin{proof}
  Let $S$ be an optimal solution to the graph $G$. For $1 \leq i \leq
  p$, let $S_i$ be the subset of $S$ such that each pair in $S_i$ has
  both its vertices in $V_i$. As noted above, the set $S_i$ is a
  solution to the graph $G[V_i]$, which imples $\omega(G[V_i]) \leq
  \pi(S_i)$. Thus,
 \[ \sum_{i=1}^p \omega(G[V_i]) \leq \sum_{i=1}^p \pi(S_i)
   \leq \pi(S) = \omega(G). \]

On the other hand, if we remove all edges in $E_{\cal P}$, and for
each $i$, apply an optimal solution $S_i'$ to the induced subgraph
$G[V_i]$, we will obviously end up with a union of disjoint cliques.
Therefore, these operations make a solution to the graph $G$ whose
weight is $\pi(E_{\cal P}) + \sum_{i=1}^p \pi(S_i') = \pi(E_{\cal
P}) + \sum_{i=1}^p \omega(G[V_i])$. This gives immediately
$\omega(G) \leq \pi(E_{\cal P}) + \sum_{i=1}^p \omega(G[V_i])$.
\qed
\end{proof}
\medskip

Lemma~\ref{lem:cutting} directly implies the following corollaries.
First, if there is no edge between two different parts in the vertex
partition ${\cal P}$, then Lemma~\ref{lem:cutting} gives

\begin{corollary}\label{lem:components}
Let $G$ be a graph with connected components $G_1$, $\dots$, $G_p$,
then $\omega(G) = \sum_{i=1}^p \omega(G_i)$, and every optimal
solution to the graph $G$ is a union of optimal solutions to the
subgraphs $G_1$, $\ldots$, $G_p$.
\end{corollary}

When $p = 2$, i.e., the vertex partition is ${\cal P} = \{X,
\overline{X} \}$, the edge set $E_{\cal P}$ becomes the cut $E(X,
\overline{X})$, and $\pi(E(X, \overline{X})) = \gamma(X)$.
Lemma~\ref{lem:cutting} gives

\begin{corollary}\label{lem:2cutting}
Let $X \subseteq V$ be a vertex set, then $\omega(G[X]) +
\omega(G[\overline{X}]) \leq \omega(G) \leq \omega(G[X]) +
\omega(G[\overline{X}]) + \gamma(X) $.
\end{corollary}

\begin{corollary}\label{lem:no-more-than-cut}
Let $G$ be a graph, and let $S^*$ be an optimal solution to $G$. For
any subset $X$ of vertices in $G$, if we let $S^*(X, \overline{X})$
be the subset of pairs in which one vertex is in $X$ and the other
vertex is in $\overline{X}$, then $\pi(S^*(X, \overline{X})) \leq
\gamma(X)$.
\end{corollary}

\begin{proof}
The optimal solution $S^*$ can be divided into three disjoint parts:
the subset $S^*(X)$ of pairs in which both vertices are in $X$, the
subset $S^*(\overline{X})$ of pairs in which both vertices are in
$\overline{X}$, and the subset $S^*(X, \overline{X})$ of pairs in
which one vertex is in $X$ and the other vertex is in
$\overline{X}$. By Corollary~\ref{lem:2cutting},
\[
\omega(G) = \pi(S^*(X)) + \pi(S^*(\overline{X})) + \pi(S^*(X,
\overline{X})) \leq \omega(G[X]) + \omega(G[\overline{X}]) +
\gamma(X).
\]
Since $\pi(S^*(X)) \geq \omega(G[X])$ and $\pi(S^*(\overline{X})) \geq
\omega(G[\overline{X}])$, we get immediately $\pi(S^*(X,
\overline{X})) \leq \gamma(X)$.
\qed
\end{proof}
\medskip

Corollary~\ref{lem:no-more-than-cut} can be informally described as
``cut preferred'' principle, which is fundamental for this problem.
Similarly we have the following lemmas.

\begin{lemma}\label{lem:edge-cutting}
Let $X$ be a subset of vertices in a graph $G$, and let $S^*$ be any
optimal solution to $G$. Let $S^*(V, \overline{X})$ be the set of
pairs in $S^*$ in which at least one vertex is in $\overline{X}$.
Then $\omega(G) \geq \omega(G[X]) + \pi(S^*(V, \overline{X}))$.
\end{lemma}

\begin{proof}
The optimal solution $S^*$ is divided into two disjoint parts: the
subset $S^*(X)$ of pairs in which both vertices are in $X$, and the
subset $S^*(V, \overline{X})$ of pairs in which at least one vertex
is in $\overline{X}$. The set $S^*(X)$ is a solution to the induced
subgraph $G[X]$. Therefore, $\pi(S^*(X)) \geq \omega(G[X])$. This
gives
\[ \omega(G) = \pi(S^*) = \pi(S^*(X)) + \pi(S^*(V, \overline{X})
    \geq \omega(G[X]) + \pi(S^*(V, \overline{X})),  \]
which proves the lemma.
\qed
\end{proof}

\begin{lemma}
Let $X$ be a subset of vertices in a graph $G$, and let $B_X$ be the
set of vertices in $X$ that are adjacent to vertices in
$\overline{X}$. Then for any optimal solution $S^*$ to $G$, if we
let $S^*(B_X)$ be the set of pairs in $S^*$ in which both vertices
are in $B_X$, then $\omega(G) + \pi(S^*(B_X)) \geq \omega(G[X]) +
\omega(G[\overline{X} \cup B_X])$.
\end{lemma}

\begin{proof}
Again, the optimal solution $S^*$ can be divided into three disjoint
parts: the subset $S^*(X)$ of pairs in which both vertices are in
$X$, the subset $S^*(\overline{X})$ of pairs in which both vertices
are in $\overline{X}$, and the subset $S^*(X, \overline{X})$ of
pairs in which one vertex is in $X$ and the other vertex is in
$\overline{X}$. We also denote by $S^*(B_X, \overline{X})$ the
subset of pairs in $S^*$ in which one vertex is in $B_X$ and the
other vertex is in $\overline{X}$. Since $S^*(X)$ is a solution to
the induced subgraph $G[X]$, we have
 \begin{eqnarray*}
 \omega(G) + \pi(S^*(B_X))
  & = & \pi(S^*(X)) + \pi(S^*(\overline{X}))
  + \pi(S^*(X, \overline{X})) + \pi(S^*(B_X)) \\
  & \geq & \omega(G[X]) + \pi(S^*(\overline{X})) +
  \pi(S^*(X, \overline{X})) + \pi(S^*(B_X))\\
  & \geq & \omega(G[X]) + \pi(S^*(\overline{X})) +
  \pi(S^*(B_X, \overline{X})) + \pi(S^*(B_X)).
 \end{eqnarray*}
The last inequality is because $B_X \subseteq X$, so $S^*(B_X,
\overline{X}) \subseteq S^*(X, \overline{X})$. Since $S' =
S^*(\overline{X}) \cup S^*(B_X, \overline{X}) \cup S^*(B_X)$ is the
subset of pairs in $S^*$ in which both vertices are in the induced
subgraph $G[\overline{X} \cup B_X]$, $S'$ is a solution to the
induced subgraph $G[\overline{X} \cup B_X]$. This gives
\[ \pi(S') = \pi(S^*(\overline{X})) + \pi(S^*(B_X, \overline{X}))
   + \pi(S^*(B_X)) \geq \omega(G[\overline{X} \cup B_X]), \]
which implies the lemma immediately.
\qed
\end{proof}
\medskip

The above results that reveal the relations between the structures
of the {\sc cluster editing} problem and graph edge cuts not only
form the basis for our kernelization results presented in the
current paper, but also are of their own importance and interests.

\section{The kernelization algorithm}\label{sec:2k}

Obviously, the number of different vertices included in a solution
$S$ of $k$ vertex pairs to a graph $G$ is upper bounded by $2k$.
Thus, if we can also bound the number of vertices that are not
included in $S$, we get a kernel. For such a vertex $v$, the clique
containing $v$ in $G \triangle S$ must be $G[N[v]]$. Inspired by
this, our approach is to check the closed neighborhood $N[v]$ for
each vertex $v$.

The observation is that if an induced subgraph (e.g. the closed
neighborhood of a vertex) is very ``dense inherently'', while is
also ``loosely connected to outside'', (\emph{i.e.}  there are very
few edges in the cut of this subgraph), it might be cut off and
solved separately.  By the cutting lemmas, the size of a solution
obtained as such should not be too far away from that of an optimal
solution. Actually, we will figure out the conditions under which
they are equal.

The subgraph we are considering is $N[v]$ for some vertex $v$. For
the connection of $N[v]$ to outside, a good measurement is
$\gamma(N[v])$. Thus, here we only need to define the density.  A
simple fact is that the fewer edges missing, the denser the subgraph
is. Therefore, to measure the density of $N[v]$, we define the
\emph{deficiency $\delta(v)$} of $N[v]$ as the total weight of
anti-edges in $G[N[v]]$, which is formally given by  $\delta(v) =
\pi( \{\{x,y\} \mid x, y \in N(v), [x, y] \not\in E\})$.

Suppose that $N[v]$ forms a single clique with no other vertices in
the resulting graph $G \triangle S$. Then anti-edges of total weight
$\delta(v)$ have to be added to make $N[v]$ a clique, and edges of
total weight $\gamma(N[v])$ have to be deleted to make $N[v]$
disjoint.  Based on this we define the \emph{stable cost} of a
vertex $v$ as $\rho(v) = 2\delta(v) + \gamma(N[v])$, and we say
$N[v]$ is \emph{reducible} if $\rho(v) < |N[v]|$.


\begin{lemma}\label{lem:non-separable}
For any vertex $v$ such that $N[v]$ is reducible, there is an
optimal solution $S^*$ to $G$ such that the vertex set $N[v]$ is
entirely contained in a single clique in the graph $G \triangle
S^*$.
\end{lemma}

\begin{proof}
Let $S$ be an optimal solution to the graph $G$, and pick any vertex
$v$ such that $N[v]$ is reducible, i.e., $\rho(v) < |N[v]|$. Suppose
that $N[v]$ is not entirely contained in a single clique in $G
\triangle S$, i.e., $N[v] = X \cup Y$, where $X \neq \emptyset$ and
$Y \neq \emptyset$, such that $Y$ is entirely contained in a clique
$C_1$ in $G \triangle S$ while $X \cap C_1 = \emptyset$ (note that
we do not assume that $X$ is in a single clique in $G \triangle S$).

Inserting all missing edges between vertices in $N[v]$ will
transform the induced subgraph $G[N[v]]$ into a clique. Therefore,
$\omega(G[N[v]]) \leq \delta(v)$. Combining this with
Corollary~\ref{lem:2cutting}, we get
\begin{eqnarray}
 \omega(G)& \leq &
  \omega(G[N[v]]) + \omega(G[\overline{N[v]}]) + \gamma(N[v]) \nonumber\\
  & \leq & \delta(v) + \omega(G[\overline{N[v]}]) + \gamma(N[v])
     \label{eq1} \\
          & = & \omega(G[\overline{N[v]}]) + \rho(v)  - \delta(v).
          \nonumber
\end{eqnarray}

Let $S(V, N[v])$ be the set of pairs in the solution $S$ in which at
least one vertex is in $N[v]$, and let $S(X, Y)$ be the set of pairs
in $S$ in which one vertex is in $X$ and the other vertex is in $Y$.
Also, let $P(X,Y)$ be the set of all pairs $(x, y)$ such that $x \in
X$ and $y \in Y$. Obviously, $\pi(S(V, N[v])) \geq \pi(S(X, Y))$
because $X \subseteq V$ and $Y \subseteq N[v]$. Moreover, since the
solution $S$ places the sets $X$ and $Y$ in different cliques, $S$
must delete all edges between $X$ and $Y$. Therefore $S(X, Y)$ is
exactly the set of edges in $G$ in which one end is in $X$ and the
other end is in $Y$. Also, by the definition of $\delta(v)$ and
because both $X$ and $Y$ are subsets of $N[v]$, the sum of the
weights of all anti-edges between $X$ and $Y$ is bounded by
$\delta(v)$. Thus, we have $\pi(S(X, Y)) + \delta(v) \geq \pi(P(X,
Y))$. Now by Lemma~\ref{lem:edge-cutting},
\begin{eqnarray}
    \omega(G)& \geq & \omega(G[\overline{N[v]}]) + \pi(S (V,N[v]))
      \nonumber\\
    & \geq & \omega(G[\overline{N[v]}]) + \pi(S(X, Y)) \label{eq2} \\
    & \geq & \omega(G[\overline{N[v]}]) + \pi(P(X, Y)) - \delta(v).
      \nonumber
  \end{eqnarray}

Combining (\ref{eq1}) and (\ref{eq2}), and noting that the weight of
each vertex pair is at least $1$, we get
\begin{equation}
  |X||Y| \leq \pi(P(X, Y)) \leq \rho(v) < |N[v]| = |X| + |Y|.
  \label{eq4}
\end{equation}
This can hold true only when $|X| = 1$ or $|Y| = 1$. In both cases,
we have $|X| \cdot |Y| = |X| + |Y| - 1$. Combining this with
(\ref{eq4}), and noting that all the quantities are integers, we
must have
\[ \pi(P(X, Y))= \rho(v), \]
which, when combined with
(\ref{eq1}) and (\ref{eq2}), gives
\begin{equation}
  \omega(G) = \omega(G[\overline{N[v]}])+\rho(v)-\delta(v) =
  \omega(G[\overline{N[v]}]) + \gamma(N[v]) + \delta(v). \label{eq3}
\end{equation}
Note that $\gamma(N[v])+\delta(v)$ is the minimum cost to insert
edges into and delete edges from the graph $G$ to make $N[v]$ a
disjoint clique. Therefore, Equality (\ref{eq3}) shows that if we
first apply edge insert/delete operations of minimum weight to make
$N[v]$ a disjoint clique, then apply an optimal solution to the
induced subgraph $G[\overline{N[v]}]$, then we have an optimal
solution $S^*$ to the graph $G$. This completes the proof of the
lemma because the optimal solution $S^*$ has the vertex set $N[v]$
entirely contained in a single clique in the graph $G \triangle
S^*$. \qed
\end{proof}
\medskip

Based on Lemma~\ref{lem:non-separable}, we have the following
reduction rule:
\begin{kernelrule}\label{rule:deficient}
  For a vertex $v$ such that $N[v]$ is reducible, insert edges
  between anti-edges in $G[N[v]]$ to make $G[N[v]]$ a clique,
  and decrease $k$ accordingly.
\end{kernelrule}


After Step~\ref{rule:deficient}, the induced subgraph $G[N[v]]$
becomes a clique with $\delta(v) = 0$ and $\rho(v) = \gamma(N[v])$.
Now we use the following rule to remove the vertices in $N(N[v])$
that are loosely connected to $N[v]$ (recall that $N(N[v])$ is the
set of vertices that are not in $N[v]$ but adjacent to some vertices
in $N(v)$, and that for two vertex subsets $X$ and $Y$, $E(X, Y)$
denotes the set of edges that has one end in $X$ and the other end
in $Y$).

\begin{kernelrule}\label{rule:multi-neighbor}
  Let $v$ be a vertex such that $N[v]$ is reducible on which
  Step~\ref{rule:deficient} has been applied.  For each vertex $x$
  in $N(N[v])$, if $\pi(E(x, N(v))) \leq |N[v]|/2$, then
  delete all edges in $E(x, N(v))$ and decrease $k$ accordingly.
\end{kernelrule}

We say that a reduction step $R$ is {\it safe} if after edge
operations of cost $c_R$ by the step, we obtain a new graph $G'$
such that the original graph $G$ has a solution of weight bounded by
$k$ if and only if the new graph $G'$ has a solution of weight
bounded by $k - c_R$.

\begin{lemma}
Step~\ref{rule:multi-neighbor} is safe.
\end{lemma}

\begin{proof}
  By Lemma~\ref{lem:non-separable}, there is an optimal solution $S$
  to the graph $G$ such that $N[v]$ is entirely contained in a single
  clique $C$ in the graph $G \triangle S$. We first prove, by
  contradiction, that the clique $C$ containing $N[v]$ in the graph
  $G \triangle S$ has at most one vertex in $\overline{N[v]}$.
  Suppose that there are $r$ vertices $u_1$, $\dots$, $u_r$ in
  $\overline{N[v]}$ that are in $C$, where $r \geq 2$.  For $1 \leq i
  \leq r$, denote by $c_i$ the total weight of all edges between $u_i$
  and $N[v]$, and by $c_i'$ the total weight of all pairs (both edges
  and anti-edges) between $u_i$ and $N[v]$. Note that $c_i' \geq
  |N[v]|$ and $\sum_{i=1}^r c_i \leq \gamma(N[v])$. Then in the
  optimal solution $S$ to $G$, the total weight of the edges inserted
  between $N[v]$ and $\overline{N[v]}$ is at least
\begin{eqnarray*}
\sum_{i=1}^r (c'_i - c_i)
 & \geq & \sum_{i=1}^r (|N[v]| - c_i)
  =    r|N[v]| - \sum_{i=1}^r c_i \\
 & \geq & r|N[v]| - \gamma(N[v])
  \geq 2|N[v]| - \gamma(N[v]) \\
 & > & 2|N[v]| - |N[v]| = |N[v]|
 > \gamma(N[v]),
\end{eqnarray*}
where we have used the fact $|N[v]| > \gamma(N[v])$ (this is because
by the conditions of the step, $\rho(v) = 2 \delta(v) + \gamma(N[v])
< |N[v]|$). But this contradicts
Corollary~\ref{lem:no-more-than-cut}.

Therefore, there is at most one vertex $x$ in $N(N[v])$ that is in
the clique $C$ containing $N[v]$ in the graph $G \triangle S$. Such
a vertex $x$ must satisfy the condition $\pi(E(x, N(v))) >
|N[v]|/2$: otherwise deleting all edges in $E(x, N(v))$ would result
in a solution that is at least as good as the one that inserts all
missing edges between $x$ and $N[v]$ and makes $N[v] \cup \{x\}$ a
clique. Thus, for a vertex $x$ in $N(N[v])$ with $\pi(E(x, N(v)))
\leq |N[v]|/2$, we can always assume that $x$ is not in the clique
containing $N[v]$ in the graph $G \triangle S$. In consequence,
deleting all edges in $E(x, N(v))$ for such a vertex $x$ is safe.
\qed
\end{proof}
\medskip

The structure of $N[v]$ changes after the above steps. The result
can be in two possible cases: (1) no vertex in $N(N[v])$ survives,
and $N[v]$ becomes an isolated clique -- then by
Corollary~\ref{lem:components}, we can simply delete the clique; and
(2) there is one vertex $x$ remaining in $N(N[v])$ (note that there
cannot be more than one vertices surviving -- otherwise it would
contradict the assumption $\gamma(N[v]) \leq \rho(v) < |N[v]|$). In
case (2), the vertex set $N[v]$ can be divided into two parts $X =
N[v] \cap N(x)$ and $Y = N[v]\backslash X$. From the proofs of the
above lemmas, we are left with only two options: either
disconnecting $X$ from $x$ with edge cost $c_X$, or connecting $Y$
and $x$ with edge cost $c_Y$. Obviously $c_X > c_Y$. Since both
options can be regarded as connection or disconnection between the
vertex set $N[v]$ and the vertex $x$, we can further reduce the
graph using the following reduction step:

\begin{kernelrule}\label{rule:pendent}
  Let $v$ be a vertex such that $N[v]$ is reducible on which
  Steps~\ref{rule:deficient} and \ref{rule:multi-neighbor} have been
  applied.  If there still exists a vertex $x$ in $N(N[v])$,
  then merge $N[v]$ into a single vertex $v'$, connect $v'$ to $x$
  with weight $c_X - c_Y$, set weight of each anti-edge between $v'$
  and other vertex to $+\infty$, and decrease $k$ by $c_Y$.
\end{kernelrule}

The correctness of this step immediately follows from above
argument.

Note that the conditions for all the above steps are only checked
once.  If they are satisfied, we apply all three steps one by one,
or else we do nothing at all.  So they are actually the parts of a
single reduction rule presented as follows:

\paragraph{The Rule.}\label{rule:the-single-rule}
  Let $v$ be a vertex satisfying $2\delta(v) + \gamma(N[v]) < |N[v]|$,
  then:
  \begin{enumerate}
\vspace{-2mm}
  \item add edges to make $G[N[v]]$ a clique and decrease $k$
    accordingly;
    \vspace{-2mm}
  \item for each vertex $x$ in $N(N[v])$ with $\pi(E(x, N[v]))
    \leq |N[v]|/2$, remove all edges in $E(x, N[v])$ and decrease $k$
    accordingly;
    \vspace{-2mm}
  \item if a vertex $x$ in $N(N[v])$ survives, merge $N[v]$
    into a single vertex (as described above) and decrease $k$ accordingly.
  \end{enumerate}

Now the following lemma implies Theorem~\ref{thm:2k} directly.

\begin{lemma}\label{lem:bound}
  If an instance of the weighted {\sc cluster editing} problem reduced
  by our reduction rule has more than $2k$ vertices, it has no
  solution of weight $\leq k$.
\end{lemma}

\begin{proof}
We divide the cost of inserting/deleting a pair $\{u,v\}$ into two
halves and assign them to $u$ and $v$ equally. Thereafter we count
the costs on all vertices.

For any two vertices with distance $2$, at most one of them is not
shown in a solution $S$: otherwise they would have to belong to the
same clique in $G \triangle S$ because of their common neighbors but
the edge between them is missing. Thus, if we let $\{v_1, v_2, \dots,
v_r \}$ be the vertices not shown in $S$, then each two of their
closed neighbors $\{N[v_1], N[v_2], \dots, N[v_r] \}$ are either the
same (when they are in the same simple series module) or mutually
disjoint.  The cost in each $N[v_i]$ is $\delta(v_i) +
\gamma(N[v_i])/2 = \rho(v_i)/2$, which is at least $|N[v_i]| / 2$,
because by our reduction rule, in the reduced instance we have
$\rho(v) \geq |N[v]|$ for each vertex $v$. Each of the vertices not in
any of $N[v_i]$ is contained in at least one pair of $S$ and therefore
bears cost at least $1/2$. Summing them up, we get a lower bound for
the total cost at least $|V|/2$. Thus, if the solution $S$ has a
weight bounded by $k$, then $k \geq |V| / 2$, i.e., the graph has at
most $2k$ vertices. \qed
\end{proof}

\section{On unweighted and real-weighted versions}
\label{sec:unweighted}

We now show how to adapt the algorithm in the previous section to
support unweighted and real-weighted versions.  Only slight
modifications are required. Therefore, the proof of the correctness
of them is omitted for the lack of space.

\paragraph{Unweighted version.} The kernelization algorithm presented
does not work for unweighted version.  The trouble arises in
Step~\ref{rule:pendent}, where merging $N[v]$ is not a valid
operation in an unweighted graph.  Fortunately, this can be easily
circumvented, by replacing Step~\ref{rule:pendent} by the following
new rule:

\addtocounter{kernelrule}{-1}
\begin{kernelrule}[U]\label{rule:pendent-unweighted}
  Let $v$ be a vertex such that $N[v]$ is reducible on which
  Steps~\ref{rule:deficient} and \ref{rule:multi-neighbor} have been
  applied. If there still exists a vertex $x$ in $N(N[v])$, then
  replace $N[v]$ by a complete subgraph $K_{|X| - |Y|}$, and connect
  $x$ to all vertices of this subgraph.
\end{kernelrule}

The correctness of this new rule is similar to the arguments in last
section, and it is easy to check the first two rules apply for the
unweighted version.  Moreover, the proof of Lemma~\ref{lem:bound}
can be easily adapted with the new rule.

\paragraph{Real-weighted version.} There are even more troubles when
weights are allowed to be real numbers, instead of only positive
integers.  The first problem is that, without the integrality,
(\ref{eq4}) cannot imply (\ref{eq3}).  This is fixable by changing
the definition of reducible closed neighborhood from $\rho(v) <
|N(v)|$ to $\rho(v) \leq |N(v)| - 1$ (they are equivalent for
integers), then (\ref{eq4}) becomes
\begin{equation}
  |X||Y| \leq \pi(P(X, Y)) \leq \rho(v) \leq |N[v]| - 1 = |X| + |Y| - 1.
  \label{eq4-new}
\end{equation}
Formulated on reducible closed neighborhood,
Steps~\ref{rule:deficient} and \ref{rule:multi-neighbor} remain the
same.

The second problem is Step~\ref{rule:pendent}, in which we need to
maintain the validity of weights.  Recall that we demand all weights
be at least $1$ for weight functions. This, although trivially holds
for integral weight functions, will be problematic for real weight
functions.  More specifically, in Step~\ref{rule:pendent}, the edge
$[x,v']$ could be assigned a weight $c_X - c_Y < 1$ when $c_X$ and
$c_Y$ differ by less than $1$.  This can be fixed with an extension
of Step~\ref{rule:pendent}:

\addtocounter{kernelrule}{-1}
\begin{kernelrule}[R]\label{rule:pendent-real-weighted}
  Let $v$ be a vertex such that $N[v]$ is reducible and that on which
  Steps~\ref{rule:deficient} and \ref{rule:multi-neighbor} have been
  applied.  If there still
  exists a vertex $x$ in $N(N[v])$, then
  \begin{itemize}
  \item if $c_X - c_Y \geq 1$, merge $N[v]$ into a single vertex $v'$,
    connect $v'$ to $x$ with weight $c_X - c_Y$, set weight of each
    anti-edge between $v'$ and other vertex to $+\infty$, and decrease
    $k$ by $c_Y$;
\vspace{-2mm}
  \item if $c_X - c_Y < 1$, merge $N[v]$ into two vertices $v'$ and
    $v''$, connect $v'$ to $x$ with weight $2$, and $v'$ to $v''$ with
    weight $2 - (c_X - c_Y)$, set weight of each anti-edge between
    $v', v''$ to other vertex to $+\infty$, and decrease $k$ by $c_X -
    2$.
  \end{itemize}
\end{kernelrule}

The new case is just to maintain the validity of the weight, and does
not make a real difference from the original case.  However, there
does exist one subtlety we need to point out, that is, the second case
might increase $k$ slightly, and this happens when $c_X - 2 \leq 0$,
then we are actually increase $k$ by $2 - c_X$.  We do not worry about
this trouble due to both theoretical and practical reasons.
Theoretically, the definition of kernelization does not forbid
increasing $k$, and we refer readers who feel uncomfortable with this
to the monographs \cite{downey-fellows,flum-grohe,niedermeier}.
Practically, 1) it will not really enlarge or complicate the graph,
and therefore any reasonable algorithms will work as the same; 2) this
case will not happen too much, otherwise the graph should be very
similar to a star, and easy to solve; 3) even using the original value
of $k$, our kernel size is bounded by $3k$.

The proof of Lemma~\ref{lem:bound} goes almost the same, with only
the constant slightly enlarged.  Due to the relaxation of the
condition of reducible closed neighborhood from $\rho(v) < |N(v)|$
to $\rho(v) \leq |N(v)| - 1$, the number of vertices in the kernel
for real-weighted version is bounded by $2.5k$.

\section{Discussion}\label{sec:discussion}

One very interesting observation is that for the unweighted version,
by the definition of simple series modules, all of the following are
exactly the same:
\[
  N[u] = N[M],\quad
  \delta(u) = \delta(M), \quad
  \mbox{and} \quad \gamma(N[u]) = \gamma(N[M]),
\]
where $M$ is the simple series module containing vertex $u$, and
$\delta(M)$ is a natural generalization of definition $\delta(v)$.
Thus it does not matter we use the module or any vertex in it, that
is, every vertex is a full representative for the simple series
module it lies in.  Although there has been a long list of linear
algorithms for finding modular decomposition for an undirected graph
(see a comprehensive survey by de Montgolfier
\cite{de-Montgolfier-thesis}), it is very time-comsuming because the
big constant hidden behind the big-O \cite{tedder-08}, and
considering that the modular decopmosition needs to be
re-constructed after each iteration, this will be helpful. It is
somehow surprising that the previous kernelization algorithms can be
significantly simplified by avoiding modular decomposition. Being
more suprising, this enables our approach to apply for the weighted
version, because one major weakness of modular decomposition is its
inability in handling weights.

One similar problem on inconsistant information is the {\sc feedback
  vertex set on tournament (fast)} problem, which asks the reverse of
minimum number of arcs to make a tournament transtive.  Given the
striking resemblances between {\sc cluster editing} and {\sc fast},
and a series of ``one-stone-two-birds'' approximation algorithms
\cite{ailon-08-ranking-and-clustering,zuylen-09} which only take
advantage of the commonalities between them, we are strongly attempted
to compare the results of these two problems from the parameterized
aspect.

For the kernelization, our result already matches the best kernel,
$(2+\epsilon) k$ for weighted {\sc fast} of Bessy et
al. \cite{fomin-09-kernel-fast}, which is obtained based on a
complicated PTAS \cite{MathieuS07}.

For the algorithms, Alon et al. \cite{ALS09} managed to generalize
the famous color coding approach to give a subexponential FPT
algorithm for {\sc fast}.  This is the first subexponential FPT
algorithm out of bidimensionality theory, which was a systematic way
to obtain subexponential algorithms, and has been intensively
studied.  This is an exciting work, and opens a new direction for
further work.  Indeed, immediately after the appearance of
\cite{ALS09}, for unweighted version, Feige reported an improved
algorithm \cite{feige-09-faster-fast} that is far simpler and uses
pure combinatorial approach.  Recently, Karpinski and Schudy reached
the same result for weighted version
\cite{schudy-10-fast-rank-aggregation-between}.  Based on their
striking resemblances, we conjecture that there is also a
subexponential algorithm for the {\sc cluster editing} problem.



\small
\begin{thebibliography}{10}

\bibitem{ailon-08-ranking-and-clustering}
Ailon, N., Charikar, M., Newman, A.:
\newblock Aggregating inconsistent information: ranking and clustering.
\newblock {\em J. ACM} 55(5), Article 23, 1-27 (2008)

\bibitem{ALS09}
Alon, N., Lokshtanov, D., Saurabh, S.:
\newblock Fast {FAST}.
\newblock In: {\em ICALP}, LNCS vol. 5555, pp. 49-58. Springer (2009)

\bibitem{bansal-04-correlation-clustering}
Bansal, N., Blum, A., Chawla, S.:
\newblock Correlation clustering.
\newblock {\em Machine Learning} 56(1), 89-113 (2004)

\bibitem{berkhin-06-survey}
Berkhin, P.:
\newblock A survey of clustering data mining techniques.
\newblock In: {\em Grouping Multidimensional Data}, Springer Berlin Heidelberg, pp. 25-71 (2006)

\bibitem{fomin-09-kernel-fast}
Bessy, S., Fomin, F. V., Gaspers, S., Paul, C., Perez, A., Saurabh, S., Thomass{\'e}, S.:
\newblock Kernels for feedback arc set in tournaments.
\newblock In: {\em CoRR}, abs/0907.2165 (2009)

\bibitem{bocker-09-experiment-clustering}
B\"ocker, S., Briesemeister, S., Klau, G. W.:
\newblock Exact algorithms for cluster editing: evaluation and experiments.
\newblock {\em Algorithmica, in press}.

\bibitem{bocker-09-weighted-clustering}
B\"ocker, S., Briesemeister, S., Bui, Q. B. A., Truss, A.:
\newblock Going weighted: parameterized algorithms for cluster editing.
\newblock {\em Theoretical Computer Science} 410, 5467-5480 (2009).

\bibitem{charikar-04-approximate-clustering}
Charikar, M., Guruswami, V., Wirth, A.:
\newblock Clustering with qualitative information.
\newblock {\em Journal of Computer and System Sciences} 71(3), 360-383 (2005)

\bibitem{chen-kernel-cluster-editing}
Chen, J., Meng, J.:
\newblock A $2k$ kernel for the cluster editing problem.
\newblock In: {\em COCOON} 2010, LNCS vol. 6196, pp. 459-468, Springer (2010)

\bibitem{chen-03-phylogenetic}
Chen, Z.-Z., Jiang, T., Lin, G.:
\newblock Computing phylogenetic roots with bounded degrees and errors.
\newblock {\em Siam J. Comp.} 32(4), 864-879 (2003)

\bibitem{CE-80-decomposition}
Cunningham, W. H., Edmonds, J.:
\newblock A combinatorial decomposition theory.
\newblock {\em Canad. J. Math.} 32(3), 734-765 (1980)

\bibitem{dean-99-www}
Dean, J., Henzinger, M. R.:
\newblock Finding related pages in the World Wide Web.
\newblock {\em Computer Networks} 31, 1467-1479 (1999)

\bibitem{downey-fellows}
Downey, R. G., Fellows, M. R.:
\newblock {\em Parameterized Complexity}, Springer (1999)

\bibitem{feige-09-faster-fast}
Feige, U.:
\newblock Faster FAST.
\newblock In: {\em CoRR}, abs/0911.5094 (2009)

\bibitem{fellows-kernel-cluster-editing}
Fellows, M. R., Langston, M. A., Rosamond, F. A., Shaw, P.:
\newblock Efficient parameterized preprocessing for cluster editing.
\newblock In {\em FCT}, LNCS vol. 4639, pp. 312-321, Springer (2007)

\bibitem{flum-grohe}
Flum, J., Grohe, M.:
\newblock {\em Parameterized Complexity Theory}, Springer (2006)

\bibitem{gramm-kernel-cluster-editing}
Gramm, J., Guo, J., H\"uffner, F., Niedermeier, R.:
\newblock Graph-modeled data clustering: exact algorithms for clique
 generation.
\newblock {\em Theory of Computing Systems} 38(4), 373-392 (2005)

\bibitem{guo-kernel-cluster-editing}
Guo J.:
\newblock A more effective linear kernelization for cluster editing.
\newblock {\em Theor. Comput. Sci.} 410(8-10), 718-726 (2009)

\bibitem{hearst-96-ir}
Hearst, M. A., Pedersen, J. O.:
\newblock Reexamining the cluster hypothesis: scatter/gather on
retrieval results.
\newblock In: {\em Proceedings of SIGIR}, pp. 76-84 (1996)

\bibitem{schudy-10-fast-rank-aggregation-between}
Karpinski, M., Schudy, W.:
\newblock Faster algorithms for feedback arc set tournament, Kemeny
rank aggregation and betweenness tournament.
\newblock In: {\em CoRR}, abs/1006.4396 (2010)

\bibitem{MathieuS07}
Kenyon-Mathieu, C., Schudy, W.:
\newblock How to rank with few errors.
\newblock In: {\em ACM Symposium on Theory of Computing (STOC)}, pp. 95-103 (2007)

\bibitem{mohring-decomposition-84}
M\"ohring, R. H., Radermacher, F. J.:
\newblock Substitution decomposition for discrete structures and
               connections with combinatorial optimization.
\newblock {\em Ann. Discrete Math.}, 19(95), 257-355, North-Holland mathematics studies, (1984)

\bibitem{de-Montgolfier-thesis}
de Montgolfier, F.:
\newblock {\em D\'ecomposition modulaire des
  graphes. Th\'eorie,extensions et algorithmes}.
\newblock Th\'ese de doctorat, Universit\'e Montpellier II (2003)

\bibitem{niedermeier}
Niedermeier, R.:
\newblock {\em Invitation to Fixed-Parameter Algorithms}, Oxford
University Press (2006)

\bibitem{shamir-04-graph-modification-problem}
Shamir, R., Sharan, R., Tsur, D.:
\newblock Cluster graph modification problems.
\newblock {\em Discrete Appl. Math.} 144 (1-2), 173-182 (2004)

\bibitem{tedder-08}
Tedder, M.; Corneil, D.; Habib, M.; Paul, C.:
\newblock Simpler linear-time modular decomposition via recursive
factorizing permutations.
\newblock In: {\em ICALP}, LNCS vol. 5125, pp. 634-645. Springer-Verlag (2008)

\bibitem{zuylen-09}
van Zuylen, A., Williamson, D. P.:
\newblock Deterministic pivoting algorithms for constrained ranking
and clustering problems.
\newblock {\em Mathematics of Operations Research} 34(3), 594-620 (2009)

\end{thebibliography}
\end{document}